\documentclass[11pt]{amsart}
\usepackage{amsmath,amssymb}
\newtheorem{theorem}{Theorem}
\newtheorem{proposition}[theorem]{Proposition}

\newtheorem{conjecture}[theorem]{Conjecture}
\theoremstyle{definition}
\newtheorem{definition}[theorem]{Definition}
\newtheorem{remark}[theorem]{Remark}
\begin{document} 

\title[Penrose inequality]{A Gibbons-Penrose inequality for surfaces in Schwarzschild spacetime}
\author{Simon Brendle and Mu-Tao Wang}
\begin{abstract}
We propose a geometric inequality for two-dimensional spacelike surfaces in the Schwarzschild spacetime. This inequality implies the Penrose inequality for collapsing dust shells in general relativity, as proposed by Penrose and Gibbons. We prove that the inequality holds in several important cases. 
\end{abstract}
\address{Department of Mathematics \\ Stanford University \\ Stanford, CA 94305}
\address{Department of Mathematics \\ Columbia University \\ 2990 Broadway \\ New York, NY 10027}
\thanks{The first author was supported in part by the National Science Foundation under grant DMS-1201924. The second author was supported in part by the National Science Foundation under grant DMS-1105483. The authors would like to thank Gary Gibbons, Gerhard Huisken, Marc Mars, and Shing-Tung Yau for helpful discussions. In particular, we are grateful to Gerhard Huisken for pointing out the existence of umbilical slices in the Schwarzschild spacetime. The second author would like to acknowledge the hospitality of National Center for Theoretical Sciences (Mathematics Division, Taipei Office) where part of this work was done during his visit.}
\maketitle 

\section{Introduction}

In \cite{Penrose}, Penrose proposed an inequality as a natural consequence of Cosmic Censorship. The original set-up of Penrose consists of a shell of dust collapsing at the speed of light. The null hypersurface swept by the incoming shell separates the spacetime into two components with the flat Minkowski metric inside. Outside the null shell, the metric is no longer flat. The spacetime is vacuum except for a delta distribution of the energy-momentum tensor of matter density supporting along the null hypersurface. 

The Penrose inequality in this case reduces to a geometric inequality on a marginally trapped surface in the null hypersurface. The location and geometry of the marginally trapped surface depends on the matter density which can be arbitrarily prescribed. This inequality should hold for a general spacelike $2$-surface in the Minkowski spacetime with minimal convexity assumptions to guarantee the regularity of the null hypersurface at infinity. It was observed by Gibbons \cite{Gibbons1} that the inequality is exactly the classical Minkowski inequality when the $2$-surface lies in an Euclidean hyperplane. Tod \cite{Tod1, Tod2} studied the case when the $2$-surface lies in the past null cone of a point and derived it from the Sobolev inequality. 

The classical Minkowski inequality was generalized to a mean convex and star-shaped surface using the method of inverse mean curvature flow (cf. \cite{Guan-Li}). Very recently, Huisken \cite{Huisken} showed that the assumption that $\Sigma$ is star-shaped can be replaced by the assumption that $\Sigma$ is outward-minimizing. 

In \cite{Gibbons2}, Gibbons proposed a reduction scheme to approach the Penrose inequality for general surfaces in the Minkowski spacetime. The idea is to project the $2$-surface orthogonally onto an Euclidean hyperplane and to relate to the Minkowski inequality of the projected surface. However, Gibbons' calculation contains a mistake and the validity of this inequality for general surfaces remains open, see Section 7.1 of \cite{Mars}. See also the detailed description in \cite{Mars-Soria}, where the Penrose inequality in the Minkowski spacetime is proven
for a large class of surfaces. 

In \cite{Wang-Yau1, Wang-Yau2}, the authors made use of  Gibbons' projection procedure in their definition of quasi-local mass. It turns out the term that is missing from Gibbons's calculation corresponds a connection $1$-form of the normal bundle with respect to a certain normal frame on the $2$-surface. This term does not vanish in general and is essential in the new definition of quasi-local mass in \cite{Wang-Yau1, Wang-Yau2}.

In \cite{Brendle-Hung-Wang}, a Minkowski type inequality for surfaces in the Anti-deSitter-Schwarzschild manifold was proved using the inverse mean curvature flow and a new Heintze-Karcher type inequality in \cite{Brendle}. When the mass parameter is zero, this inequality implies the Penrose inequality for a $2$-surface that lies on the hyperbola of the Minkowski spacetime, see Section 8 of \cite{Wang}. In this article, we propose a conjecture generalizing the Penrose inequality for surfaces in the Minkowski spacetime. More specifically, the ambient space is the Schwarzschild spacetime. We  prove that the inequality holds in following four cases: (1) when the surface lies in a totally geodesic time slice; (2) when the surface lies in a totally umbilical slice; (3) when the surface lies in a null hypersurface emanating from a sphere of symmetry; (4) when the surface lies in a convex static timelike hypersurface (see Definition \ref{convex.static} for a precise statement).

We remark that the Riemannian Penrose inequality was proved by Huisken-Ilmanen \cite{Huisken-Ilmanen} and Bray \cite{Bray}. For other related work on the Penrose inequality, we refer to \cite{Mars} and references therein.

\section{Statement of the Penrose inequality}

\subsection{Minkowski spacetime}
Let $\Sigma$ be a two-dimensional spacelike closed  embedded orientable surface in the Minkowksi space $\mathbb{R}^{3,1}$. Throughout the article, we assume $\Sigma$ is diffeomorphic to $S^2$. We consider a fixed future timelike vector $T_0 \in \mathbb{R}^{3,1}$ satisfying $\langle T_0,T_0 \rangle = -1$. 

We recall the mean curvature vector field $\vec{H}$ of $\Sigma$, which is the unique normal vector field such that the variation of area of $\Sigma$ in a normal variation field $V$ is given by $-\int_\Sigma \langle \vec{H},V \rangle \, d\mu$. The convention we adopt here makes the mean curvature vector of a standard round sphere inward pointing.  Let $L$ and $\underline{L}$ be two null normals of $\Sigma$ with $\langle L, \underline{L}\rangle=2$. We assume $L$ is future-directed and $\underline{L}$ is past-directed (both outward pointing whenever this makes sense). In terms of $L $ and $\underline{L}$, we have 
\[\vec{H} = \frac{1}{2} \, (\langle \vec{H}, L\rangle \, \underline{L} + \langle \vec{H}, \underline{L}\rangle \, L).\]

The dual mean curvature vector $\vec{J}$ is defined as 
\[\vec{J} = \frac{1}{2} \, (\langle \vec{H}, L \rangle \, \underline{L} - \langle \vec{H}, \underline{L} \rangle \, L). \] 
$\vec{J}$ satisfies $\langle \vec{J}, \vec{J} \rangle = -\langle \vec{H},\vec{H} \rangle$ and $\langle \vec{J}, \vec{H} \rangle = 0$. In fact, $\vec{J}$ is uniquely characterized by these properties, up to a sign. The choice here makes $\vec{J}$ a future timelike vector in case $\vec{H}$ is inward spacelike.

The following inequality for spacelike $2$-surfaces in the Minkowski spacetime was proposed by Penrose in \cite{Penrose}: 

\begin{conjecture} (Penrose)
Suppose that $\Sigma$ is past null convex in the sense that the past null hypersurface generated by $\Sigma$ is smooth. Then 
\begin{equation} 
\label{penrose.inequality}
- \int_\Sigma \langle \vec{J},T_0 \rangle \, d\mu \geq \sqrt{16\pi \, |\Sigma|}. 
\end{equation}
\end{conjecture}

By the divergence theorem, we have 
\[\int_\Sigma \langle \vec{H},T_0 \rangle \, d\mu = 0.\] 
This implies 
\begin{equation}
-\int_\Sigma \langle \vec{H},L \rangle \, \langle \underline{L},T_0 \rangle \, d\mu = \int_\Sigma \langle \vec{H},\underline{L} \rangle \, \langle L,T_0 \rangle \, d\mu = -\int_\Sigma \langle \vec{J},T_0 \rangle \, d\mu.
\end{equation}
Thus, the inequality \eqref{penrose.inequality} can be rewritten as 
\begin{equation} 
-\int_\Sigma \langle \vec{H},L \rangle \, \langle \underline{L},T_0 \rangle \, d\mu \geq \sqrt{16\pi \, |\Sigma|}. 
\end{equation}
This formulation is independent of the choice of $L$ and $\underline{L}$ except the normalization $\langle L,\underline{L} \rangle = 2$. If we choose $\underline{L}$ such that $\langle \underline{L},T_0 \rangle=1$, then the inequality \eqref{penrose.inequality} can be written in the form 
\begin{equation} 
\label{null.exp}
\int_\Sigma \theta \, d\mu \geq \sqrt{16\pi \, |\Sigma|}, 
\end{equation}
where $\theta = -\langle \vec{H},L \rangle$ corresponds to the future (outward) null expansion. Note that the inequality \eqref{null.exp} is equivalent to the inequality (51) in \cite{Mars}.

\subsection{Schwarzschild spacetime}

The Schwarzschild spacetime metric is given by 
\begin{equation}
\label{Sch_coor}
-(1-\frac{2m}{r}) \, dt^2 + \frac{1}{1-\frac{2m}{r}} \, dr^2 + r^2 \, g_{S^2}, 
\end{equation}
where $g_{S^2} = d\theta^2 + \sin^2\theta \, d\phi^2$ is the round metric on $S^2$.

Let $\Sigma$ be a closed embedded orientable spacelike $2$-surface in the Schwarzschild spacetime. Let $L$ and $\underline{L}$ be two null normals of $\Sigma$ with $\langle L,\underline{L} \rangle = 2$. Again we assume $L$ is future-directed and $\underline{L}$ is past-directed. 

Since $\frac{\partial}{\partial t}$ is a Killing field, we have 
\[\int_\Sigma \langle \vec{H},\frac{\partial}{\partial t} \rangle \, d\mu = 0.\]
As above, this implies 
\[-\int_\Sigma \langle \vec{H},L \rangle \, \langle \underline{L},\frac{\partial}{\partial t} \rangle \, d\mu = \int_\Sigma \langle \vec{H},\underline{L} \rangle \, \langle L,\frac{\partial}{\partial t} \rangle \, d\mu = -\int_\Sigma \langle \vec{J},\frac{\partial}{\partial t} \rangle \, d\mu,\]
where 
\[\vec{J} = \frac{1}{2} \, (\langle \vec{H}, L \rangle \, \underline{L} - \langle \vec{H}, \underline{L} \rangle \, L)\] 
denotes the dual mean curvature vector.

\begin{conjecture}
\label{conjecture} 
Let $\Sigma$ be a spacelike $2$-surface in the Schwarzschild spacetime. Suppose that the past null hypersurface generated by $\Sigma$ is smooth. Then 
\begin{equation}
\label{Sch.penrose1} 
-\int_\Sigma \langle \vec{J}, \frac{\partial}{\partial t} \rangle \, d\mu + 16\pi m \geq \sqrt{16\pi|\Sigma|}. 
\end{equation} 
Here $m$ is the total mass of the Schwarzschild spacetime.
\end{conjecture}

Of course, an equivalent formulation is 
\begin{equation} 
\label{Sch.penrose2} 
-\int_\Sigma \langle \vec{H},L \rangle \, \langle \underline{L},\frac{\partial}{\partial t} \rangle \, d\mu + 16\pi m \geq \sqrt{16\pi|\Sigma|}. 
\end{equation}
The Schwarzschild spacetime belongs to the larger class of static spacetimes. We recall that a spacetime $S$ is \textit{static} if the metric is of the form $-\Omega^2 \, dt^2+g_M$ where $g_M$ is a Riemannian metric on a 3-manifold $M$ and $\Omega$ is a smooth function defined on $M$. 

\begin{definition}
\label{convex.static}
Let $B$ be a complete timelike hypersurface in a static spacetime $S$. We say that $B$ is \textit{convex static} if $B =\{(t, x) :  t\in \mathbb{R}, x\in \hat{\Sigma} \}$ for some 2-surface $\hat{\Sigma} \subset M$, and the second fundamental form  $\hat{h}_{ab}$ and the induced metric $\hat{g}_{ab}$ of $\hat{\Sigma}$ in $M$ satisfies $\hat{h}_{ab} \geq \Omega^{-1} \, \hat{\nu}(\Omega) \, \hat{g}_{ab} > 0$. Here, $\hat{\nu}$ denotes the outward-pointing unit normal to $\hat{\Sigma}$.
\end{definition}

The condition $\hat{h}_{ab} \geq \Omega^{-1} \, \hat{\nu}(\Omega) \, \hat{g}_{ab} > 0$ has a natural geometric interpretation: it implies that the second fundamental form $\Pi$ of the timelike hypersurface $B$ is nonnegative when evaluated at a null vector, i.e. $\Pi(X, X)\geq 0$ if $X$ is null and tangent to $B$. 

In this paper, we prove that the inequality holds for a large class of spacelike $2$-surfaces in the Schwarzschild spacetime.

\begin{theorem} 
Let $\Sigma$ be a closed embedded orientable spacelike $2$-surface in the Schwarzschild spacetime. The Gibbons-Penrose inequality \eqref{Sch.penrose1} holds in the following cases:
\begin{enumerate}
\item $\Sigma$ lies in a totally geodesic spacelike hypersurface and $\Sigma$ is mean convex and star-shaped.
\item $\Sigma$ lies in a totally umbilical (spherically symmetric) spacelike hypersurface and $\Sigma$ is mean convex and star-shaped.
\item $\Sigma$ lies in a outward directed null hypersurface emanating from a sphere of symmetry.
\item $\Sigma$ lies in a convex static timelike hypersurface. 
\end{enumerate}
\end{theorem}

We remark that by taking $m=0$, these give rise to the Penrose inequality in the Minkowski spacetime in the corresponding cases (see also \cite{Wang}).

\section{Proof of the inequality in four cases}

\subsection{Surfaces in a totally geodesic time slice}
We first check the case when $\Sigma$ lies in a totally geodesic time-slice ($t=0$) and thus the induced metric is 
\[\frac{1}{1-\frac{2m}{r}} \, dr^2+r^2 \, g_{S^2}.\] 
The future timelike unit normal is given by $e_0 = \frac{1}{\sqrt{1-\frac{2m}{r}}} \, \frac{\partial}{\partial t}$. Let $L = e_0+\nu$ and $\underline{L} = -e_0+\nu$ be the two null normals where $\nu$ is the outward unit normal of $\Sigma$ in the time-slice. We compute $\vec{H} = -H\nu$ and $\vec{J} = He_0$. This gives 
\[-\int_\Sigma \langle \vec{J},\frac{\partial}{\partial t} \rangle \, d\mu = -\int_\Sigma H \, \langle e_0,\frac{\partial}{\partial t} \rangle \, d\mu = \int_\Sigma H \, \sqrt{1-\frac{2m}{r}} \, d\mu.\] 
Thus, the inequality \eqref{Sch.penrose1} in this case is equivalent to 
\begin{equation} 
\label{eq_totally_geodesic} 
\int_\Sigma H \sqrt{1-\frac{2m}{r}} \, d\mu + 16\pi m \geq \sqrt{16\pi|\Sigma|}. 
\end{equation} 
Notice that the horizon area $|\partial M|$ equals $16\pi m^2$, and the static potential for the Schwarzschild space-time is $\sqrt{1-\frac{2m}{r}}$. Hence, the inequality \eqref{eq_totally_geodesic} follows from results in \cite{Brendle-Hung-Wang}. (Theorem 1 in \cite{Brendle-Hung-Wang} works for arbitrary negative cosmological constant, and the result needed here follows by sending the cosmological constant to $0$.)

\subsection{Surfaces in a totally umbilical slice}

We claim that the inequality in Theorem 1 of \cite{Brendle-Hung-Wang} for surfaces in the Anti-deSitter-Schwarzschild manifold corresponds to inequality \eqref{Sch.penrose1} for surfaces in a spherically symmetric umbilical slice of the Schwarzschild space-time. Let us recall the definition of the three-dimensional Anti-deSitter-Schwarzschild manifold \footnote{The definition here is slightly different from \cite{Brendle} and \cite{Brendle-Hung-Wang}, as we use $2m$ as the mass parameter instead of $m$}. We fix a real number $m > 0$, and let $s_0$ denote the unique positive solution of the equation 
\begin{equation} 
\label{eq_s_0} 
1 - 2m \, s_0^{-1} + \lambda^2 s_0^2 = 0. 
\end{equation} 
We then consider the manifold $M = S^2 \times [s_0,\infty)$ equipped with the Riemannian metric 
\[\bar{g} = \frac{1}{1 - 2m \, s^{-1} + \lambda^2 s^2} \, ds^2 + s^2 \, g_{S^2},\]
where $g_{S^2}$ is the standard round metric on the unit sphere $S^2$. The sectional curvatures of $(M,\bar{g})$ approach $-1$ near infinity, so $\bar{g}$ is asymptotically hyperbolic. Moreover, the scalar curvature of $(M,\bar{g})$ equals $-6$. The boundary $\partial M = S^2 \times \{s_0\}$ is referred to as the horizon.

Now we return to the Schwarzschild spacetime. Consider a function $\rho=\rho(s)$ that satisfies 
\[\rho'(s) = \frac{\lambda s}{(1-\frac{2m}{s})\sqrt{1-\frac{2m}{s}+\lambda^2 s^2}}\] 
for some constant $\lambda>0$. Take the embedding of $(2m, \infty)\times S^2$ into the Schwarzschild space-time by $\Psi(s,\theta,\phi) = (\rho(s),s,\theta,\phi)$ in Schwarzschild coordinates $(t,r,\theta,\phi)$ and denote the image by $\hat{M}=\{(t,r,\theta,\phi): t=\rho(s), r=s\}$. 

Substituting $t=\rho(s)$ and $r=s$ in \eqref{Sch_coor}, it follows that the induced metric on $\hat{M}$ is given by
\[\frac{1}{1-\frac{2m}{s}+\lambda^2 s^2} ds^2 + s^2 \, g_{S^2},\] 
which is isometric to the one on an Anti-deSitter-Schwarzschild three-manifold $M$.

\begin{remark}
The function $\rho$ appears to be only defined on $(2m, \infty)$. However, we can extend $\hat{M}$ in an extension of Schwarzschild space-time so that the domain of definition of $s$ extends to $(s_0,\infty)$ where $s_0$ is the unique positive root of $1-\frac{2m}{s}+\lambda^2 s^2$. We refer to Section 6 of \cite{Mars} where such an extension is carried out by using advanced Eddington-Finkelstein coordinates. In any case, we shall denote by $M$ the one that is extended to $(s_0,\infty)$ which is referred as the Anti-deSitter-Schwarzschild manifold in \cite{Brendle-Hung-Wang}. Note that $\hat{M} \subset M$ is an isometric embedding.
\end{remark}

\begin{proposition} 
The hypersurface $\hat{M}$ is umbilical, i.e the second fundamental form is proportional to the induced metric. 
\end{proposition}

\begin{proof} 
Let $b(s)=1-\frac{2m}{s}$ and $f(s)=\sqrt{1-\frac{2m}{s}+\lambda ^2 s^2}$. We have the following relation:
\[b^{-1}-(\rho')^2 b = f^{-2}.\] 
An orthonormal coframe adapted to the hypersurface $\hat{M}$ is given by 
\[\theta^0=\frac{1}{\sqrt{b^{-1}-b(\rho')^2}} \, (dt-\rho' \, dr) = f(s) \, (dt-\rho' \, dr),\]
\[\theta^1=\frac{1}{\sqrt{b^{-1}-b(\rho')^2}} \, (b\rho' \, dt-b^{-1} \, dr) = f(s) \, (b\rho' \, dt-b^{-1} \, dr),\]
\[\theta^2=s \, d\theta,\] and
\[\theta^3=s \sin\theta \, d\phi,\] 
where $\theta^0$ is the unit conormal that is dual to the unit future timelike normal 
\begin{equation} 
\label{unit_normal} 
e_0=\frac{f(s)}{b(s)} \, \frac{\partial}{\partial t} + \lambda s \, \frac{\partial}{\partial r}. 
\end{equation}
The second fundamental form can be computed using this coframe and we derive
\[p_{11}=\frac{d}{ds}(\frac{b\rho'}{\sqrt{b^{-1}-b(\rho')^2}}),\] 
\[p_{22}=p_{33}=\frac{1}{s} \, \frac{b\rho'}{\sqrt{b^{-1}-b(\rho')^2}}.\]
We check that 
\[\frac{b\rho'}{\sqrt{b^{-1}-b(\rho')^2}}=\lambda s.\] 
Thus, $p_{11}=p_{22}=p_{33}=\lambda$ and $\hat{M}$ is umbilical.  
\end{proof}

\begin{proposition}
\label{Sigma.lies.in.an.umbilic.hypersurface}
For a spacelike $2$-surface $\Sigma$ in $\hat{M}$ that is mean convex and star-shaped, the inequality \eqref{Sch.penrose1} holds.
\end{proposition}

\begin{proof} 
We assume $\lambda=1$ for simplicity. (The general case can be reduced to this special case by scaling.) Consider a spacelike $2$-surface $\Sigma$ in the umbilical hyersurface $\hat{M}$. Let $\nu$ be the outward unit normal of $\Sigma$ in $\hat{M}$, and let $L=e_0+\nu$ and $\underline{L}=-e_0+\nu$ be the two null normals. The mean curvature vector $\vec{H}$ is given by $-H\nu + 2e_0$ where $H$ is the mean curvature of $\Sigma$ in $\hat{M}$ with respect to $\nu$. Consequently, the dual mean curvature vector is 
\[\vec{J} = \langle \vec{H},e_0 \rangle \, \nu - \langle \vec{H},\nu \rangle \, e_0 = He_0 - 2\nu.\] 
This implies 
\[-\int_\Sigma \langle \vec{J},\frac{\partial}{\partial t} \rangle \, d\mu = -\int_\Sigma H \, \langle e_0,\frac{\partial}{\partial t} \rangle \, d\mu + 2 \int_\Sigma \langle \nu,\frac{\partial}{\partial t} \rangle \, d\mu.\] 
As above, we identify the hypersurface $\hat{M}$ with a region in the three-dimensional Anti-deSitter-Schwarzschild manifold. The function 
\[-\langle e_0,\frac{\partial}{\partial t} \rangle = \theta^0(\frac{\partial}{\partial t}) = \sqrt{1-\frac{2m}{s}+\lambda^2 s^2} = f(s)\] 
is exactly the static potential for the Anti-deSitter-Schwarzschild space-time. Let us denote by $(\frac{\partial}{\partial t})^\top$ the component of $\frac{\partial}{\partial t}$ that is tangential to $\hat{M}$. From \eqref{unit_normal} and $\Psi_*(\frac{\partial }{\partial s})=\rho'(s)\frac{\partial}{\partial t}+\frac{\partial}{\partial r}$, we derive 
\[(\frac{\partial}{\partial t})^\top = -s f(s) \Psi_*(\frac{\partial}{\partial s}),\] 
hence 
\[\langle \nu,\frac{\partial}{\partial t} \rangle = -\langle \nu,s \, f(s) \, \Psi_*(\frac{\partial}{\partial s}) \rangle.\] 
Putting these facts together, we obtain 
\[-\int_\Sigma \langle \vec{J},\frac{\partial}{\partial t} \rangle \, d\mu = \int_\Sigma H \, f \, d\mu - 2 \int_\Sigma \langle \nu,s f(s) \frac{\partial}{\partial s} \rangle \, d\mu.\] The can be viewed as an equation on $M$ through the isometry. 
Recall $(M, \bar{g})$ is defined for $s\in [s_0, \infty)$ with 
\[\bar{g}=\frac{1}{f(s)^2} \, ds^2+s^2 \, g_{S^2}.\]
Applying the divergence theorem on $M$ gives
\[\int_\Sigma \langle \nu,s f(s) \frac{\partial}{\partial s} \rangle \, d\mu = \int_\Omega \text{\rm div}_{\bar{g}} (s f(s)\frac{\partial}{\partial s}) \, d\text{\rm vol} + \int_{\partial M} \langle \nu,s f(s)\frac{\partial}{\partial s}  \rangle \, d\mu\] 
where $\partial M$ is the horizon and $\Omega$ is the region enclosed by $\partial M$ and $\Sigma$.  A straightforward  computation shows  
\[\text{\rm div}_{\bar{g}}(sf \frac{\partial}{\partial s}) = 3f.\] 
In fact, $s f(s) \frac{\partial}{\partial s}$ is the conformal Killing field used in \cite{Brendle}.
On the other hand, on a level surface of $s$, $\nu=f(s) \, \frac{\partial}{\partial s}$ and $\langle \nu,sf(s)\frac{\partial}{\partial s} \rangle = s$. Taking the limit $s \searrow s_0$, we obtain 
\[\int_{\partial M} \langle \nu,sf(s)\frac{\partial}{\partial s} \rangle \, d\mu = 4\pi s_0^3.\]
In summary, we have shown that 
\[\int_\Sigma \langle \nu,s f(s) \frac{\partial}{\partial s} \rangle \, d\mu = \int_\Omega 3f \, d\text{\rm vol} + 4\pi s_0^3,\] 
hence 
\[-\int_\Sigma \langle \vec{J},\frac{\partial}{\partial t} \rangle \, d\mu = \int_\Sigma H \, f \, d\mu - \int_\Omega 6f \, d\text{\rm vol} - 8\pi s_0^3.\] 
Now  recall from \cite{Brendle-Hung-Wang} that for such a surface in the Anti-deSitter-Schwarzschild space $M$,
\[\int_\Sigma f H \, d\mu - 6\int_\Omega f \, d\text{\rm vol} \geq \sqrt{16\pi|\Sigma|}-8\pi s_0.\] 
Therefore, inequality \eqref{Sch.penrose1} follows by combining the last two inequalities and the defining equation \eqref{eq_s_0} of $s_0$,  which implies $s_0^3+s_0=2m$.
\end{proof}

\subsection{Surfaces in a null cone}
Let $\Sigma$ be a spacelike $2$-surface which is contained in the null hypersurface 
\[N = \{(t,r,\theta,\phi): t=s+2m \log (\frac{s}{2m}-1), r=s, s>2m\}.\] 
Let $L$ and $\underline{L}$ be the null normal vectors to $\Sigma$. Note that the future outward null normal $L$ is tangential to the null hypersurface $N$. Since $\Sigma$ is spacelike, $\Sigma$ can be written as a radial graph 
\[\Sigma = \{(t,r,\theta,\phi):  t=r+2m \log (\frac{r}{2m}-1),  r=u(\theta,\phi), (\theta, \phi)\in S^2\}\] 
for some function $u: S^2 \to (2m,\infty)$.

For each $\lambda>0$, we denote by $\rho_\lambda$ the unique solution of the ODE 
\[\rho'(s) = \frac{\lambda s}{(1-\frac{2m}{s})\sqrt{1-\frac{2m}{s}+\lambda^2 s^2}}\] 
such that $\rho_\lambda(4m)=4m$. It is easy to see that the functions $\rho_\lambda(s)$ converge smoothly to the function $s+2m \log(\frac{s}{2m}-1)$ as $\lambda\rightarrow \infty$ for $s$ in compact subintervals of $(2m,\infty)$. Let 
\[\hat{M}_\lambda = \{ (t, r, \theta, \phi): t=\rho_\lambda(s), r=s, s> 2m \}.\] 
We have seen above that $\hat{M}_\lambda$ is an umbilic hypersurface which is isometric to the $3$-dimensional Anti-deSitter-Schwarzschild manifold. Moreover, the surface 
\[\Sigma_\lambda = \{(t,r,\theta,\phi): t=\rho_\lambda({r}), r=u(\theta,\phi), (\theta, \phi)\in S^2\}\] 
can be viewed as a star-shaped surface within the Anti-deSitter-Schwarzschild manifold $\hat{M}_\lambda$. 

As $\lambda \to \infty$, the hypersurfaces $\hat{M}_\lambda$ converge smoothly to the null hypersurface $N$. Moreover, the surfaces $\Sigma_\lambda$ converge smoothly to the original spacelike $2$-surface $\Sigma$. In particular, the mean curvature vector of $\Sigma_\lambda$ converges to the mean curvature vector of $\Sigma$ as $\lambda \to \infty$, and the dual mean curvature vector of $\Sigma_\lambda$ converges to the dual mean curvature vector of $\Sigma$. 

Finally, the unit normal vector field to $\Sigma_\lambda$ within $\hat{M}_\lambda$ converges to the future outward normal vector $L$ after suitable rescaling. Since the null expansion of $\Sigma$ along $L$ is strictly positive, we conclude that the mean curvature of $\Sigma_\lambda$ (viewed as a hypersurface in $\hat{M}_\lambda$) is strictly positive when $\lambda$ is sufficiently large. Therefore, Proposition \ref{Sigma.lies.in.an.umbilic.hypersurface} implies that the Gibbons-Penrose inequality \eqref{Sch.penrose1} holds for $\Sigma_\lambda$ when $\lambda$ is sufficiently large. Taking the limit as $\lambda \to \infty$, we conclude that the Gibbons-Penrose inequality \eqref{Sch.penrose1} also holds for the original surface $\Sigma$.

\subsection{Surfaces in a convex static timelike hypersuface}
Let us consider a spacelike $2$-surface $\Sigma$ in the Schwarzschild spacetime. For abbreviation, we put $\hat{\Sigma} = \pi(\Sigma)$, where $\pi: (t,r,\theta,\phi) \mapsto (r,\theta,\phi)$ denotes the projection to the $t=0$ slice along the Killing vector field $\frac{\partial}{\partial t}$. Let us choose parametrizations $F$ and $\hat{F}$ for $\Sigma$ and $\hat{\Sigma}$ so that $F(x) = (\tau(x), \hat{F}(x))$. Clearly,  
\[\frac{\partial F}{\partial x_a} = \frac{\partial \hat{F}}{\partial x_a} + \frac{\partial \tau}{\partial x_a} \, \frac{\partial}{\partial t}\] 
Hence, the induced metric on $\Sigma$ is related to the metric on $\hat{\Sigma}$ by 
\[\hat{g}_{ab} = g_{ab} + f^2 \, \partial_a \tau \, \partial_b \tau,\] 
where, as usual, $f = \sqrt{1-\frac{2m}{r}}$. This gives 
\[\hat{g}^{ab} = g^{ab} - \frac{f^2 \, g^{ac} \, g^{bd} \, \partial_c \tau \, \partial_d \tau}{1 + f^2 \, |\nabla \tau|^2},\] 
where $|\nabla \tau|^2 = g^{ab} \, \partial_a \tau \, \partial_b \tau$.

We next relate the second fundamental form of $\Sigma$ to the second fundamental form of the projected surface $\hat{\Sigma}$. Consider the timelike hypersurface $B = \{(t,x): t \in \mathbb{R}, \, x \in \hat{\Sigma}\}$, and let $\nu$ denote the outward-pointing unit normal vector to $B$. We may extend $\nu$ to a vector field defined in an open neighborhood of $B$ such that $[\nu,\frac{\partial}{\partial t}] = 0$ and $\langle \nu,\frac{\partial}{\partial t} \rangle = 0$. 

Note that $\nu$ is a normal vector field along both $\Sigma$ and $\hat{\Sigma}$. Moreover, we have 
\begin{align*} 
\langle \frac{\partial F}{\partial x_a},D_{\frac{\partial F}{\partial x_b}} \nu \rangle 
&= \langle \frac{\partial \hat{F}}{\partial x_a},D_{\frac{\partial \hat{F}}{\partial x_b}} \nu \rangle + \frac{\partial \tau}{\partial x_a} \, \frac{\partial \tau}{\partial x_b} \, \langle \frac{\partial}{\partial t},D_{\frac{\partial}{\partial t}} \nu \rangle \\ 
&+ \frac{\partial \tau}{\partial x_a} \, \langle \frac{\partial}{\partial t},D_{\frac{\partial \hat{F}}{\partial x_b}} \nu \rangle + \frac{\partial \tau}{\partial x_b} \, \langle \frac{\partial \hat{F}}{\partial x_a},D_{\frac{\partial}{\partial t}} \nu \rangle \\ 
&= \langle \frac{\partial \hat{F}}{\partial x_a},D_{\frac{\partial \hat{F}}{\partial x_b}} \nu \rangle + \frac{\partial \tau}{\partial x_a} \, \frac{\partial \tau}{\partial x_b} \, \langle \frac{\partial}{\partial t},D_\nu \frac{\partial}{\partial t} \rangle \\ 
&+ \frac{\partial \tau}{\partial x_a} \, \langle \frac{\partial}{\partial t},D_{\frac{\partial \hat{F}}{\partial x_b}} \nu \rangle + \frac{\partial \tau}{\partial x_b} \, \langle \frac{\partial \hat{F}}{\partial x_a},D_\nu \frac{\partial}{\partial t} \rangle \\ 
&= \hat{h}_{ab} - \frac{\partial \tau}{\partial x_a} \, \frac{\partial \tau}{\partial x_b} \, f \, \nu(f),
\end{align*} 
where $\hat{h}_{ab}$ is the second fundamental form of the projected surface $\hat{\Sigma}$. This implies 
\begin{align*} 
-\langle \vec{H},\nu \rangle 
&= g^{ab} \, \langle \frac{\partial F}{\partial x_a},D_{\frac{\partial F}{\partial x_b}} \nu \rangle \\ 
&= g^{ab} \, \hat{h}_{ab} - |\nabla \tau|^2 \, f \, {\nu}(f) \\ 
&= \hat{g}^{ab} \, \hat{h}_{ab} + \frac{f^2 \, g^{ac} \, g^{bd} \, \partial_c \tau \, \partial_d \tau}{1 + f^2 \, |\nabla \tau|^2} \, \hat{h}_{ab} - |\nabla \tau|^2 \, f \, \nu(f) \\ 
&= \hat{H} + \frac{f^2 \, g^{ac} \, g^{bd} \, \partial_c \tau \, \partial_d \tau}{1 + f^2 \, |\nabla \tau|^2} \, (\hat{h}_{ab} - f^{-1} \, \nu(f) \, \hat{g}_{ab}), 
\end{align*} 
where $\hat{H} = \hat{g}^{ab} \, \hat{h}_{ab}$ denotes the mean curvature of $\hat{\Sigma}$. If $B$ is convex static in the sense of Definition \ref{convex.static}, then the tensor $\hat{h}_{ab} - f^{-1} \, \nu(f) \, \hat{g}_{ab}$ is positive semidefinite, and we obtain 
\[-\langle \vec{H},\nu \rangle \geq \hat{H}.\] 
On the other hand, we have 
\[-\langle \vec{J},\frac{\partial}{\partial t} \rangle = -\langle \vec{H},\nu \rangle \, \sqrt{-\langle \Big ( \frac{\partial}{\partial t} \Big )^\perp,\Big ( \frac{\partial}{\partial t} \Big )^\perp \rangle} = -\langle \vec{H},\nu \rangle \, f \, \sqrt{1 + f^2 \, |\nabla \tau|^2},\] 
where $|\nabla \tau|^2 = g^{ab} \, \partial_a \tau \, \partial_b \tau$. Putting these facts together, we obtain the pointwise inequality 
\[-\langle \vec{J},\frac{\partial}{\partial t} \rangle \geq \hat{H} \, f \, \sqrt{1 + f^2 \, |\nabla \tau|^2}.\] 
The volume elements of $\Sigma$ and $\hat{\Sigma}$ are related by $d\hat{\mu} =\sqrt{1+f^2 \, |\nabla\tau|^2} \, d\mu$. Hence, integrating the last equation gives 
\[-\int_\Sigma \langle \vec{J}, \frac{\partial}{\partial t}\rangle \, d\mu \geq \int_{\hat{\Sigma}} \hat{H} \, f \, d\hat{\mu}.\] 
On the other hand, since $B$ is convex static, the surface $\hat{\Sigma}$ is star-shaped and convex. Using our results above, we obtain 
\[\int_{\hat{\Sigma}} \hat{H} \, f \, d\hat{\mu}\geq \sqrt{16\pi|\hat{\Sigma}|} - 16\pi m.\] 
Hence, the desired inequality follows by observing that $|\hat{\Sigma}| \geq |\Sigma|$.

\end{document}